\documentclass[aip,jmp]{revtex4-1}

\usepackage{amssymb}
\usepackage{latexsym}
\usepackage{amsmath,amsthm}
\usepackage{tikz}
\usepackage{hyperref}
\usepackage{graphicx}
\usepackage{color}
\usepackage{tikz}
\usetikzlibrary{decorations.markings}

\usepackage{verbatim}

\newtheorem{thm}{Theorem}
\newtheorem{lem}{Lemma}
\newtheorem{cor}{Corollary}

\newcommand{\geql}{\geq_{\mathrm{lex}}}
\newcommand{\gtl}{>_{\mathrm{lex}}}

\newcommand{\ltl}{<_{\mathrm{lex}}}
\newcommand{\ud}{\mathrm{d}}
\newcommand{\ui}{\mathrm{i}}
\newcommand{\ue}{\mathrm{e}}
\newcommand{\UI}{\mathrm{I}}
\newcommand{\lyn}{\mathrm{Lyn}_q}

\newcommand{\COE}{\mathrm{COE}}

\newcommand{\CUE}{\mathrm{CUE}}

\newcommand{\diag}{\mathrm{diag}}

\newcommand{\la}{\langle}
\newcommand{\ra}{\rangle}

\newcommand{\ppo}{\bar{\gamma}} 
\newcommand{\po}{\tilde{\gamma}} 

\begin{document}

\global\long\def\alph{\mathcal{A}}

\global\long\def\E{\mathcal{E}}
\global\long\def\G{\mathcal{G}}
\global\long\def\V{\mathcal{V}}

\global\long\def\str#1{\mathrm{Str}_{q}(#1)}

\title{Lyndon word decompositions and pseudo orbits on $q$-nary graphs}
\date{\today}
\author{R Band}
\email{ramband@technion.ac.il}
\affiliation{Department of Mathematics, Technion - Israel Institute of Technology, Haifa 32000, Israel}
\author{JM Harrison}
\email{jon\_harrison@baylor.edu}
\author{M Sepanski}
\email{mark\_sepanski@baylor.edu}
\affiliation{Department of Mathematics, Baylor University, Waco, TX 76798, USA}

\begin{abstract}
A foundational result in the theory of Lyndon words (words that are
strictly earlier in lexicographic order than their cyclic permutations)
is the Chen-Fox-Lyndon theorem which states that every word has a
unique non-increasing decomposition into Lyndon words. This article extends
this factorization theorem, obtaining the proportion of these decompositions
that are strictly decreasing. This result is then used to count primitive
pseudo orbits (sets of primitive periodic orbits) on $q$-nary graphs.
As an application we obtain a diagonal approximation to the variance
of the characteristic polynomial coefficients of
$q$-nary quantum graphs.
\end{abstract}



\maketitle

\section{Introduction}\label{sec:intro}

A fundamental tool used to understand the combinatorics of words is
the Lyndon factorization \cite{CFL58} (see also Ref. \onlinecite{Lothaire});
every word has a unique standard decomposition into a non-increasing sequence
of Lyndon words. Lyndon words being those words that occur strictly
earlier in lexicographic order than any of their rotations. 

So, for example, the Lyndon words on the binary
alphabet with length $\leq3$ arranged in lexicographic order are 
\begin{equation}
0\ltl001\ltl01\ltl011\ltl1\ .\label{eq:Lyndon word examples}
\end{equation}
And the unique standard decompositions of the binary words of length $3$ into non increasing sequences of Lyndon words are,
\begin{align}\label{eq:Lyndon decomposition examples}
(0)(0)(0),~~\mathbf{(001),\mathbf{~~(01)(0)}},~~\mathbf{(011)},~~(1)(0)(0),~~\mathbf{(1)(01)},~~(1)(1)(0),~~(1)(1)(1).
\end{align}

The Lyndon
factorization finds applications in diverse problems from the theory of free
Lie algebras \cite{Lothaire}, to quasi-symmetric functions \cite{H01}
and data compression techniques \cite{MRRS13}. In this article we
extend this foundational result to obtain the proportion of the standard
decompositions that are strictly decreasing, so the standard decomposition has no repetitions. For words of a fixed
length on an alphabet of $q$ letters the proportion that have strictly
decreasing Lyndon factorizations is shown to be $(q-1)/q$, independent of the word length.  Returning to the example above, the strictly decreasing standard decompositions of binary words length of $3$ in (\ref{eq:Lyndon decomposition examples}) are indicated in bold, we see half, $(2-1)/2$, of the standard decompositions are strictly decreasing. 

The remainder of the article applies this new combinatorial result
to a problem in the field of quantum chaos. We focus on quantum graphs,
which are a widely studied model of quantum chaos introduced by Kottos
and Smilansky \cite{KS97,KS99}. Quantum graphs are also used in other
diverse areas of mathematical physics including Anderson localization,
microelectronics, nanotechnology, photonic crystals and superconductivity,
see Ref. \onlinecite{BerkolaikoKuchment, GS06} for an introduction. In Ref. \onlinecite{BHJ12} the authors showed
that spectral properties of quantum graphs are precisely encoded in
finite sums over collections of primitive periodic orbits called primitive
pseudo orbits. Here, we introduce graph families which we call $q$-nary
graphs, where there is a bijection between primitive pseudo orbits
on those graphs and strictly decreasing standard decompositions.

The spectrum of the graph is encoded in the characteristic polynomial of the quantum evolution operator, defined in terms of the scattering matrices at the vertices, see section \ref{characteristic polynomial}. 
It was shown in Ref. \onlinecite{BHJ12}
that the coefficients of the characteristic polynomial $a_n$ can be expressed as a sum over primitive pseudo
orbits of the graph. It is the variance of these coefficients, averaged over the spectral parameter, which
we treat for families of $q$-nary graphs.  A $q$-nary graph has vertices labeled by words of length $m$ on an alphabet of $q$ letters.  By counting the number of strictly decreasing standard decompositions we obtain a diagonal approximation for the variance,
\begin{equation}
\la|a_{n}|^{2}\ra_{\diag}=\frac{q-1}{q}\ .
\end{equation}
This can be compared to the corresponding random matrix result \cite{Hetal96},
\begin{equation}
\la|a_{n}|^{2}\ra_{\CUE}=1\ .
\end{equation}
 The grounds for such a comparison is the Bohigas-Giannoni-Schmidt
conjecture \cite{BGS84} which asserts that typically the spectrum
of a classically chaotic quantum system corresponds, in the semiclassical limit, to that of an
ensemble of random matrices determined by the symmetries of the quantum
system. 
In quantum mechanics the semiclassical limit is the limit of large energies or equivalently the limit $\hbar \to 0$.  The appropriate semiclassical limit for graphs is a limit of a sequence of graphs with increasing number of edges, which corresponds to increasing the length of the words labeling the vertices $m$.
The deviation we see from random matrix theory is consistent with previous investigations of the variance
\cite{KS99,T00,T01}.   From our result it is clear that the deviation does not vanish for a given family of $q$-nary graphs in the semiclassical limit.    However, the discrepancy would disappear for a sequence of graphs where the degree of the vertices increases, which is equivalent to increasing $q$.  This suggests that random matrix results for the variance may be recovered under stronger conditions than those typically required for other spectral properties.

The article is laid out as follows. In Section \ref{sec:introduction_to_Lyndon}
we introduce the terminology associated with Lyndon factorizations.  In Section \ref{sec:Counting-strictly-decreasing}
we count the number of strictly decreasing standard decompositions
(Theorem \ref{strictly decreasing factorization theorem}) which is
a main result of this article. In Section \ref{graphs} we introduce
$q$-nary graphs which are families of directed graphs and use Theorem
\ref{strictly decreasing factorization theorem} to count the primitive
pseudo orbits on these graphs. Section \ref{characteristic polynomial}
describes how coefficients of the graph's characteristic polynomial
can be expressed as finite sums over primitive pseudo orbits. In Section
\ref{diagonal} we apply the primitive pseudo orbit count to obtain
a diagonal approximation for the variance of coefficients of the characteristic
polynomial of $q$-nary quantum graphs and compare it to predictions
from random matrix theory.

\section{Introduction to Lyndon words\label{sec:introduction_to_Lyndon}}

\label{background} In this article we consider factorizations of
words over a totally ordered alphabet $\alph$ of $q$ letters. The
lexicographic order of words is defined in the following natural way.
Let,
\begin{align}
w & =a_{1}a_{2}\dots a_{l}\label{eq:lex defn 1}\\
w' & =b_{1}b_{2}\dots b_{k}
\end{align}
with $a_{i},b_{j}\in\alph$. Then $w\gtl w'$ iff there exists $i\leq\min\{l,k\}$
such that $a_{1}=b_{1},\dots,a_{i-1}=b_{i-1}$ and $a_{i}>b_{i}$
or $l>k$ and $a_{1}=b_{1},\dots,a_{k}=b_{k}$.

Two words, $w$ and $w'$ are said to be conjugate if $w=uv$ and
$w'=vu$ for some words $u$ and $v$. Hence, two words are conjugate
if and only if one may be obtained as a rotation (or cyclic shift)
of the other and conjugacy is clearly an equivalence relation. A word
is a \emph{Lyndon} word if it is strictly less than all other words
in its conjugacy class. So, going back to the binary example, the Lyndon words on the binary
alphabet with length $\leq4$ are 
\begin{equation}
0\ltl 0001\ltl001\ltl 0011 \ltl01\ltl011\ltl 0111 \ltl1\ .\label{eq:Lyndon word example 2}
\end{equation}
For a fixed alphabet we denote the set of Lyndon words of length $l$
by $\lyn(l)$ and $L_{q}(l)=|\lyn(l)|$. A useful classical result
involving the number of Lyndon words is the following lemma (see e.g.,
Ref. \onlinecite{Lothaire}).
\begin{lem}
\label{lem:classic}
\begin{equation}
\sum_{l\mid m}lL_{q}(l)=q^{m}
\end{equation}

\end{lem}
The lemma \cite{Lothaire}follows from the fact that every word of length $m$ is a repetition of some word $w$ of length $l |m $, where $w$ is in the conjugacy class of some Lyndon word.  There are $L_q(l)$ conjugacy classes and each conjugacy class has $l$ distinct words.

The Chen-Fox-Lyndon factorization theorem \cite{CFL58} (see also
Ref. \onlinecite{Lothaire}) is the following fundamental result in the theory
of Lyndon words.
\begin{thm}
\label{CFL theorem} Any non-empty word $w$ can be uniquely written
as a concatenation of Lyndon words in non-increasing lexicographic order,
\begin{equation}
w=v_{1}v_{2}\dots v_{k}\  ,\label{eq:standard decomposition}
\end{equation}
where each $v_{j}$ is a Lyndon word and $v_{j}\geql v_{j+1}$.
\end{thm}
We call the unique factorization (\ref{eq:standard decomposition})
the \emph{standard decomposition} (or \emph{Lyndon factorization})
of $w$.
So for example the Lyndon factorization of the word \texttt{LYNDON} is \texttt{(LYN)(DON)}.

 Furthermore, we say that a standard decomposition is \emph{strictly
decreasing} if $v_{j}\gtl v_{j+1}$ for $j=1,\dots,k-1$. We will
denote by $\str n$ the number of strictly decreasing standard decompositions
of words of length $n$ from an alphabet of $q$ letters. So, for
example, the standard decompositions of binary words of length $4$
are shown below where the strictly decreasing standard decompositions
are indicated in bold,
\begin{displaymath}
\begin{array}{cccc}
(0)(0)(0)(0), & \mathbf{(0001)}, &\mathbf{(001)(0)},&\mathbf{(0011)},\\
(01)(0)(0),&(01)(01),&\mathbf{(011)(0)},&\mathbf{(0111)}, \\
(1)(0)(0)(0), & \mathbf{(1)(001)}, &\mathbf{(1)(01)(0)},&\mathbf{(1)(011)},\\
(1)(1)(0)(0),&(1)(1)(01),&{(1)(1)(1)(0)},&(1)(1)(1)(1). \\ 
\end{array}
\end{displaymath}
We see that precisely half of the binary words of length $4$ have
strictly decreasing standard decompositions. In general, for an alphabet
of $q$ letters, the proportion of words that have strictly decreasing
standard decompositions is $(q-1)/q$, which we prove in the next
section.

\section{Counting strictly decreasing standard decompositions\label{sec:Counting-strictly-decreasing}}

The following theorem is the main combinatorial result of the paper.
\begin{thm}
For words of length $n\geq2$, \label{strictly decreasing factorization theorem}
\begin{equation}
\str n=(q-1)q^{n-1}\ .
\end{equation}

\end{thm}
We formally define a generating function for the number of strictly
decreasing standard decompositions as
\begin{equation}
p\left(x\right)=\sum_{n=0}^{\infty}\str n\cdot x^{n}\ ,
\end{equation}
where we set $\str 0=1$ and $\str 1=q$. If we also define a function,
\begin{equation}
f\left(x\right)=\frac{qx^{2}-1}{qx-1}=1+qx+\sum_{n=2}^{\infty}\left(q-1\right)q^{n-1}x^{n}\ ,
\end{equation}
then proving the theorem is equivalent to showing that $p=f$ on some
interval. To do this we use the following lemma.
\begin{lem}
\label{p lemma}
\begin{align}
p\left(x\right) & =\prod_{l=1}^{\infty}\left(1+x^{l}\right)^{L_q\left(l\right)}
\end{align}
\end{lem}
\begin{proof}
Observe that the set of words with strictly decreasing standard decomposition
is in bijection with the set of subsets of all Lyndon words. The bijection
is implemented by taking any collection of distinct Lyndon words,
arranging them in (strictly) decreasing order, and concatenating.
That this is invertible follows from the Chen-Fox-Lyndon theorem, as
every word has a unique non-increasing standard decomposition.
\end{proof}

\begin{proof} {\em (of Theorem \ref{strictly decreasing factorization theorem})}
As $p(0)=f(0)=1$ we note that $p=f$ on $(-1,1)$ if
\begin{equation}
\frac{\ud}{\ud x}\log p=\frac{\ud}{\ud x}\log f\!
\end{equation}
on $(-1,1)$.
From Lemma \ref{p lemma},
\begin{align}
\log p & =\sum_{l=1}^{\infty}L_q(l)\log\left(1+x^{l}\right)\\
 & =\sum_{l=1}^{\infty}L_q(l)\sum_{j=1}^{\infty}\frac{(-1)^{j+1}}{j}x^{lj}\ ,
\end{align}
where the second equality is valid for $\left|x\right|<1$. Hence,
\begin{align}
\frac{\ud}{\ud x}\log p & =-\frac{1}{x}\sum_{l=1}^{\infty}\sum_{j=1}^{\infty}(-1)^{j}lL_q(l)x^{lj}\ ,\\
 & =-\frac{1}{x}\sum_{m=1}^{\infty}\sum_{l|m}(-1)^{\frac{m}{l}}lL_q(l)x^{m}\ .
\end{align}
Splitting the sum over $m$ into sums over odd and even terms respectively,
\begin{align}
\frac{\ud}{\ud x}\log p & =\frac{1}{x}\sum_{m=1}^{\infty}\, \sum_{l|2m-1}lL_q(l)x^{2m-1}-\frac{1}{x}\sum_{m=1}^{\infty}\sum_{l|2m}(-1)^{\frac{2m}{l}}lL_q(l)x^{2m}\\
 & =\frac{1}{x}\sum_{m=1}^{\infty}\sum_{l|m}lL_q(l)x^{m}-\frac{1}{x}\sum_{m=1}^{\infty}\sum_{l|2m}\big(1+(-1)^{\frac{2m}{l}}\big)lL_q(l)x^{2m}\\
 & =\frac{1}{x}\sum_{m=1}^{\infty}\sum_{l|m}lL_q(l)x^{m}-\frac{2}{x}\sum_{m=1}^{\infty}\sum_{l|m}lL_q(l)x^{2m},
\end{align}
where, in the last step we used the fact that coefficients in the
second sum vanish unless $l$ divides $m$.  Applying Lemma
\ref{lem:classic},
\begin{align}
\frac{\ud}{\ud x}\log p & =\frac{1}{x}\sum_{m=1}^{\infty}(qx)^{m}-\frac{2}{x}\sum_{m=1}^{\infty}(qx^{2})^{m}\\
 & =\frac{1}{x}\frac{qx}{1-qx}-\frac{2}{x}\frac{qx^{2}}{1-qx^{2}}\ .
\end{align}

Finally comparing this to,
\begin{equation}
\frac{\ud}{\ud x}\log f=\frac{q}{1-qx}-\frac{2qx}{1-qx^{2}}\ .
\end{equation}
completes the proof.
\end{proof}

\section{Quantum $Q$-nary graphs and their pseudo-orbits}

\label{graphs}

A graph $\G$ is a set of vertices $\V$ connected by a set of edges
$\E$. We consider graphs with directed edges where each edge $e=(u,v)\in\E$,
connects an origin vertex $o(e)=u$ to a terminal vertex $t(e)=v$. We write
$e\sim v$ if $v$ is a vertex in $e$. The number of edges $e\sim v$
is $d_{v}$ the \emph{degree} of $v$. The total number of edges is
$E=|\mathcal{E}|$.

Let $q$ and $m$ be positive integers. We define a $q$-nary graph
of order $m$ in the following way. We use an alphabet, $\alph$,
of $q$ letters and let the set of graph vertices be labeled by the
$q^{m}$ words of length $m$. The edges of the graph are labeled
by words of length $q^{m+1}$ where the first $m$ letters of the
edge label designate the origin vertex and the last $m$ letters denote
the terminal vertex. Consequently every vertex of the $q$-nary graph
has $q$ incoming edges and $q$ outgoing edges. See Figure \ref{fig:binary graph}
for an example of a binary graph with $2^{3}$ vertices and Figure
\ref{fig:ternary graph} for a ternary graph with $3^{2}$ vertices.

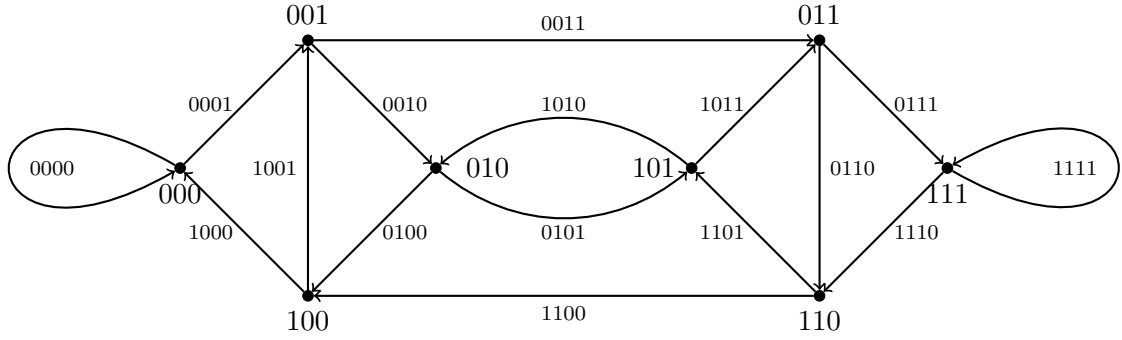
\begin{figure}[tbh]
\vspace*{8pt}

\begin{centering}
\begin{tikzpicture}[scale=1.7]
   \tikzstyle{graph node}=[draw,circle,fill=black,minimum size=4pt,inner sep=0pt]
    \draw (-3,0) node[graph node] (0) [label=-90:${000}$] {};
    \draw (-2,1) node[graph node] (1) [label=90:${001}$] {};
    \draw (-1,0) node[graph node] (2) [label=0:${\, \, \, 010}$] {};
    \draw (2,1) node[graph node] (3) [label=90:${011}$] {};
    \draw (-2,-1) node[graph node] (4) [label=-90:${100}$] {};
    \draw (1,0) node[graph node] (5) [label=180:${101}$] {};
    \draw (2,-1) node[graph node] (6) [label=-90:${110}$] {};
    \draw (3,0) node[graph node] (7) [label=-90:${111}$] {};
    \draw[thick,->] (0) -- (1) node[draw=none,fill=none,font=\scriptsize,midway,left] {$0001$};
    \draw[thick,->] (1) -- (2) node[draw=none,fill=none,font=\scriptsize,midway,right] {$0010$};
    \draw[thick,->] (1) -- (3) node[draw=none,fill=none,font=\scriptsize,midway,above] {$0011$};
    \draw[thick,->] (2) -- (4) node[draw=none,fill=none,font=\scriptsize,midway,right] {$0100$};
    \draw[thick,->] (2) to [bend right=40] (5);
    \draw[thick,->] (3) -- (6) node[draw=none,fill=none,font=\scriptsize,midway,right] {$0110$};
    \draw[thick,->] (3) -- (7) node[draw=none,fill=none,font=\scriptsize,midway,right] {$0111$};
    \draw[thick,->] (4) -- (0) node[draw=none,fill=none,font=\scriptsize,midway,left] {$1000$};
    \draw[thick,->] (4) -- (1) node[draw=none,fill=none,font=\scriptsize,midway,left] {$1001$};
    \draw[thick,->] (5) to [bend right=40] (2);
    \draw[thick,->] (6) -- (4)  node[draw=none,fill=none,font=\scriptsize,midway,below] {$1100$};
    \draw[thick,->] (6) -- (5) node[draw=none,fill=none,font=\scriptsize,midway,left] {$1101$};
    \draw[thick,->] (7) -- (6) node[draw=none,fill=none,font=\scriptsize,midway,right] {$1110$};
    \draw[thick,->] (5) -- (3) node[draw=none,fill=none,font=\scriptsize,midway,left] {$1011$};
    \path[thick,->,min distance=2cm] (0)edge[in=210,out=150]  (0);
    \path[thick,->,min distance=2cm] (7)edge[in=30,out=330] (7);
    \draw  (0,.5) node[draw=none,fill=none,font=\scriptsize] {$1010$};
    \draw  (0,-0.5) node[draw=none,fill=none,font=\scriptsize] {$0101$};
    \draw  (-4,0) node[draw=none,fill=none,font=\scriptsize] {$0000$};
    \draw  (4,0) node[draw=none,fill=none,font=\scriptsize] {$1111$};
\end{tikzpicture}
\par\end{centering}

\vspace*{8pt}
 \protect\protect\caption{A binary graph with $2^{3}$ vertices.}
\label{fig:binary graph}
\end{figure}

\begin{figure}
\vspace*{8pt}

\begin{centering}
\begin{tikzpicture}[scale=2]
   \tikzstyle{graph node}=[draw,circle,fill=black,minimum size=4pt,inner sep=0pt]
    \draw (-2,0) node[graph node] (0) [label=0:${00}$] {};
    \draw (-1,1) node[graph node] (1) [label=135:${01}$] {};
    \draw (0,-2) node[graph node] (2) [label=-90:${02}$] {};
    \draw (-1,-1) node[graph node] (3) [label=225:${10}$] {};
    \draw (0,0) node[graph node] (4) [label=90:${11}$] {};
    \draw (1,1) node[graph node] (5) [label=45:${12}$] {};
    \draw (0,2) node[graph node] (6) [label=90:${20}$] {};
    \draw (1,-1) node[graph node] (7) [label=-45:${21}$] {};
    \draw (2,0) node[graph node] (8) [label=180:${22}$] {};
    \draw[thick,->] (0) -- (1) node[draw=none,fill=none,font=\scriptsize,midway,left] {$001$};
    \draw[thick,->] (3) -- (0) node[draw=none,fill=none,font=\scriptsize,midway,left] {$100$};
    \draw[thick,->] (1) -- (4) node[draw=none,fill=none,font=\scriptsize,midway,right] {$011$};
    \draw[thick,->] (1) -- (5) node[draw=none,fill=none,font=\scriptsize,midway,above] {$012$};
    \draw[thick,->] (6) -- (1) node[draw=none,fill=none,font=\scriptsize,midway,left] {$201$};
    \draw[thick,->] (5) -- (6) node[draw=none,fill=none,font=\scriptsize,midway,right] {$120$};
    \draw[thick,->] (4) -- (3) node[draw=none,fill=none,font=\scriptsize,midway,right] {$110$};
    \draw[thick,->] (3) -- (2) node[draw=none,fill=none,font=\scriptsize,midway,left] {$102$};
    \draw[thick,->] (1) to [bend right=20] (3);
    \draw[thick,->] (3) to [bend right=20] (1);
    \draw[thick,->] (5) to [bend right=20] (7);
    \draw[thick,->] (7) to [bend right=20] (5);
    \draw[thick,->] (4) -- (5)  node[draw=none,fill=none,font=\scriptsize,midway,left] {$112$};
    \draw[thick,->] (5) -- (8) node[draw=none,fill=none,font=\scriptsize,midway,right] {$122$};
    \draw[thick,->] (8) -- (7) node[draw=none,fill=none,font=\scriptsize,midway,right] {$221$};
    \draw[thick,->] (7) -- (3) node[draw=none,fill=none,font=\scriptsize,midway,below] {$210$};
    \draw[thick,->] (7) -- (4) node[draw=none,fill=none,font=\scriptsize,midway,left] {$211$};
    \draw[thick,->] (2) -- (7) node[draw=none,fill=none,font=\scriptsize,midway,right] {$021$};
    \path[thick,->,min distance=1cm] (0)edge[in=210,out=150]  (0);
    \path[thick,->,min distance=1cm] (8)edge[in=30,out=330] (8);
    \path[thick,->,min distance=1cm] (4)edge[in=-70,out=-110] (4);
    \draw[thick,->] (6) to [bend right=40] (0);
    \draw[thick,->] (8) to [bend right=40] (6);
    \draw[thick,->] (2) to [bend right=40] (8);
    \draw[thick,->] (0) to [bend right=40] (2);
   \path[thick,->, distance=4cm] (6)edge[in=180,out=180]  (2);
    \path[thick,->, distance=4cm] (2)edge[in=0,out=0]  (6);
    \draw  (-3.25,0) node[draw=none,fill=none,font=\scriptsize] {$202$};
    \draw  (3.25,0) node[draw=none,fill=none,font=\scriptsize] {$020$};
    \draw  (-2.5,0.25) node[draw=none,fill=none,font=\scriptsize] {$000$};
    \draw  (2.5,0.25) node[draw=none,fill=none,font=\scriptsize] {$222$};
    \draw  (0,-0.85) node[draw=none,fill=none,font=\scriptsize] {$111$};
    \draw  (-1.35,0) node[draw=none,fill=none,font=\scriptsize] {$010$};
    \draw  (1.35,0) node[draw=none,fill=none,font=\scriptsize] {$212$};
    \draw  (-0.65,0) node[draw=none,fill=none,font=\scriptsize] {$101$};
    \draw  (0.65,0) node[draw=none,fill=none,font=\scriptsize] {$121$};
    \draw  (-1.5,1.5) node[draw=none,fill=none,font=\scriptsize] {$200$};
    \draw  (-1.5,-1.5) node[draw=none,fill=none,font=\scriptsize] {$002$};
    \draw  (1.5,-1.5) node[draw=none,fill=none,font=\scriptsize] {$022$};
    \draw  (1.5,1.5) node[draw=none,fill=none,font=\scriptsize] {$220$};
\end{tikzpicture}
\par\end{centering}

\vspace*{8pt}
 \protect\caption{A ternary graph with $3^{2}$ vertices.}
\label{fig:ternary graph}
\end{figure}
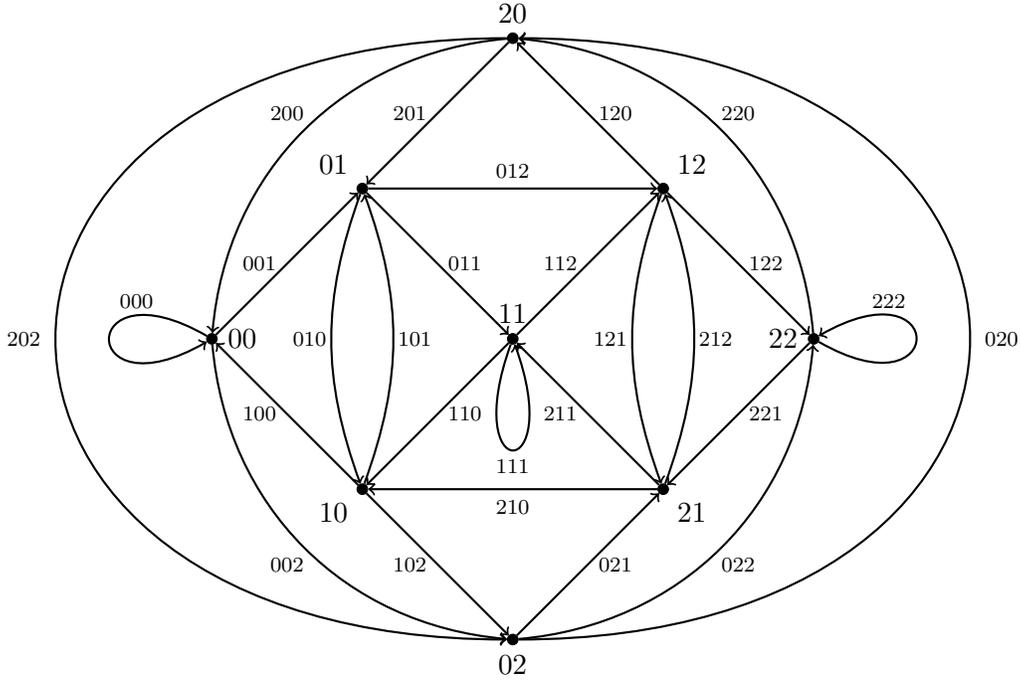

A \emph{path} $p=(v_{1},v_{2},\dots,v_{l+1})$ of topological length
$E_{p}=l$ can be labeled by a sequence of $l+1$ connected vertices
or alternatively by the corresponding connected $l$ edges $p=(e_{1},\dots,e_{l})$
where $e_{j}=(v_{j},v_{j+1})$. On the $q$-nary graph with $q^{m}$
vertices a path of length $l$ is be labeled by a word $w=a_{1},\dots,a_{l+m}$
where the connected vertices on the path are obtained by reading off
consecutive subwords of $m$ letters; so the first vertex is labeled
by $a_{1},\dots,a_{m}$ the second by $a_{2},\dots,a_{m+1}$ and so
on. Clearly every $q$-nary graph is connected as any vertex can be
reached from any other vertex by a path of at most $m$ edges. On
the other hand, a \emph{periodic orbit} $\gamma=(v_{1},\dots,v_{l},v_{1})$
of $E_{\gamma}=l$ edges, which is a closed path on $\G$, is labeled
by a word $w=a_{1},\dots,a_{l}$ of length $l$, where to obtain all
the $l$ subwords of length length $m$ the letters of $w$ are rotated
cyclically. For example, in Figure \ref{fig:binary graph} the periodic
orbit of topological length $1$ denoted by $0$ corresponds to the
loop $0000$ joining vertex $000$ to itself. Alternatively $0001$
is the periodic orbit of topological length $4$,
\[
000\to001\to010\to100\to000\ .
\]
Clearly the number of periodic orbits of length $l$ on a $q$-nary
graph is $q^{l}$. A \emph{primitive periodic orbit} is a periodic
orbit that is not a repetition of a shorter periodic orbit. We observe
that there is a bijection between the primitive periodic orbits and
Lyndon words. Indeed, a Lyndon word serves as a representative of
its conjugacy class and by definition cannot be a repetition of a
shorter word (see Section \ref{sec:introduction_to_Lyndon}).

A \emph{pseudo orbit} $\po=\{\gamma_{1},\dots,\gamma_{M}\}$ on $\G$
is a set of periodic orbits. We will use $m_{\po}=M$ to denote the
number of periodic orbits in the pseudo orbit. The topological length
of the pseudo orbit is
\begin{equation}
E_{\po}=\sum_{j=1}^{M}E_{\gamma_{j}}\ .
\end{equation}
A \emph{primitive pseudo orbit} $\ppo=\{\gamma_{1},\dots,\gamma_{M}\}$
is a set of primitive periodic orbits in which no periodic orbit appears
more than once; so primitive pseudo orbits omit repetitions of periodic
orbits both in the collection of periodic orbits and inside each periodic
orbit that makes up the pseudo orbit.

We can now see that there is a bijection between primitive pseudo
orbits and words whose standard decomposition does not contain any
Lyndon word more than once so that they are strictly decreasing. Thus,
we can apply Theorem \ref{strictly decreasing factorization theorem}
to count primitive pseudo orbits on $q$-nary graphs.
\begin{cor}
\label{cor: no. of primitive pseudo orbits} The number of primitive
pseudo orbits of topological length $n$ on a $q$-nary graph of order
$m$ is $(q-1)q^{n-1}$.
\end{cor}

\section{The characteristic polynomial of a quantum graph}

\label{characteristic polynomial}

Quantum graphs were introduced as a model system in which to study
spectral properties when the corresponding classical dynamics is chaotic
by Kottos and Smilansky \cite{KS97,KS99}. Spectral properties of
quantum binary graphs ($q$-nary graphs with $q=2$) were investigated
by Tanner in Ref. \onlinecite{T00}.

A discrete graph can be turned into a metric graph by associating
a length $l_{e}>0$ to each edge $e\in\E$. There are two main approaches
to quantize a metric graph which are closely related (see Refs. \onlinecite{BerkolaikoKuchment,GS06}).
We describe the approach adopted here. Given an arbitrary
directed metric graph, where each vertex has $q$ incoming and $q$ outgoing
edges, we equip each vertex, $v\in\V$ with a prescribed unitary $q\times q$
matrix. We call this matrix a vertex scattering matrix and denote
it by $\sigma^{(v)}$. Each entry of this matrix $\sigma_{e,e'}^{(v)}$
is a scattering transmission amplitude from edge $e'$
to edge $e$. Hence, the entries, $\sigma_{e,e'}^{(v)}$ are indexed
such that $e'$ is an edge directed towards the vertex $v$ and $e$
is directed out of it, so $v=t(e')=o(e)$.

In particular we consider $q\times q$ vertex scattering matrices of the form,
\begin{equation}
\sigma_{e,e'}^{(v)}=\frac{1}{\sqrt{q}}\left(\begin{array}{ccccc}
1 & 1 & 1 & \dots & 1\\
1 & \omega & \omega^{2} & \dots & \omega^{q-1}\\
1 & \omega^{2} & \omega^{4} & \dots & \omega^{2(q-1)}\\
\vdots & \vdots & \vdots & \ddots & \vdots\\
1 & \omega^{k-1} & \omega^{2(k-1)} & \dots & \omega^{(q-1)(q-1)}
\end{array}\right)\ ,\label{eq:DFT_vertex_scattering}
\end{equation}
where $\omega=\ue^{\frac{2\pi\ui}{q}}$ is a primitive $q$-th root of
unity.   This matrix is the Discrete Fourier Transform (DFT)
matrix.   Such vertex scattering matrices have the advantage of being democratic, in the sense that the transmission probability $|\sigma_{e,e'}^{(v)}|^2=1/q$ for every outgoing edge $e$.  Consequently graphs with the DFT vertex scattering matrices are a well studied model of quantum chaos for which spectral properties are seen to converge rapidly to the predictions of random matrix theory.\cite{T01}

All the vertex scattering matrices $\sigma^{(v)}$ can be combined
into a single $E\times E$ unitary matrix in the following way.  Fixing an arbitrary order for the $E$ graph edges we compose an
$E\times E$ matrix, $\Sigma$, whose entries are indexed by the graph
edges and set
\[
\Sigma_{e,e'}=\begin{cases}
\sigma_{e,e'}^{(v)} & ~~v=t(e')=o(e)\\
0 & ~~\textrm{otherwise}
\end{cases},
\]
where $t(e')$ marks the terminal vertex of $e'$ and $o(e)$ marks
the origin vertex of $e$.

Next we define $L=\textrm{diag}\{l_{1},\dots,l_{E}\}$
to be a diagonal matrix of all edge lengths and set $U\left(k\right)=\ue^{\ui kL}\Sigma$,
which is called the unitary (or quantum) evolution operator. The graph
spectrum is then defined as
\begin{equation}
\left\{ \left.k^{2}\right|\det\left(\UI-U\left(k\right)\right)=0\right\} \ .\label{eq:secular_cond_tmp}
\end{equation}
This is the set of eigenvalues of the negative Laplacian on the metric graph when the vertex scattering matrices are those obtained from a self-adjoint realization of the operator, see e.g. Ref. \onlinecite{BerkolaikoKuchment}.

The \emph{characteristic polynomial} of $U\left(k\right)$ is,
\begin{equation}
F_{\xi}\left(k\right)=\det\left(\xi\UI-U\left(k\right)\right)=\sum_{n=0}^{2B}a_{n}\xi^{2B-n}\ ,\label{eq:characteristic_polynomial}
\end{equation}
and we note that the graph's eigenvalues are obtained as the zeros
of $F_{\xi=1}$. The characteristic polynomial coefficients, $a_{n}$,
are the spectral quantity which we investigate here. It is shown in Ref. \onlinecite{BHJ12} that
each $a_{n}$ may be expressed as a sum over pseudo orbits on the
graph in the following way. To each periodic orbit $\gamma=(e_{1},\dots,e_{m})$
on the quantum graph it is natural to associate a metric length,
\begin{equation}
l_{\gamma}=\sum_{j=1}^{m}l_{e_{j}}
\end{equation}
and a stability amplitude, the product of the elements of the scattering
matrix around the orbit,
\begin{equation}
A_{\gamma}=\Sigma_{e_{2}e_{1}}\Sigma_{e_{3}e_{2}}\dots \Sigma_{e_{n}e_{n-1}}\Sigma_{e_{1}e_{m}}\ .\label{eq:defn A_gamma}
\end{equation}
Then a pseudo orbit $\po=\{\gamma_{1},\dots,\gamma_{M}\}$ acquires
a metric length and stability amplitude,
\begin{align}
l_{\po} & =\sum_{j=1}^{M}l_{\gamma_{j}}\ ,\\
A_{\po} & =\prod_{j=1}^{M}A_{\gamma_{j}}\ .
\end{align}
In Ref. \onlinecite{BHJ12} the authors prove the following theorem.
\begin{thm}
\label{thm:a_n_expansion} The coefficients of the characteristic
polynomial $F_{\xi}\left(k\right)$ are given by
\begin{equation}
a_{n}=\sum_{\ppo|\,E_{\ppo}=n}\left(-1\right)^{m_{\ppo}}A_{\ppo}\left(k\right)\exp\left(\mathrm{i}kl_{\ppo}\right)\ ,\label{eq:a_n_expansion}
\end{equation}
where the (finite) sum is over all the primitive pseudo orbits of
topological length $n$.
\end{thm}
In Ref. \onlinecite{BHJ12} Theorem \ref{thm:a_n_expansion} is used to express
the secular function, zeta function and spectral determinant in terms
of the dynamical properties of finite numbers of pseudo orbits. In
the final section we present results on the second moment
of the coefficients.

\section{Variance of coefficients of the characteristic polynomial}

\label{sec:coefficients_of_char_poly} \label{diagonal}


Typically, spectral properties of quantum chaotic systems can
be modeled by the spectrum of a corresponding ensemble of random matrices
according to the conjecture of Bohigas, Giannoni and Schmidt \cite{BGS84}.
The variance of the coefficients of the characteristic polynomial
of an $E\times E$ random scattering matrix from the Circular Orthogonal
Ensemble (COE) and Circular Unitary Ensemble (CUE) \cite{Hetal96}
are,
\begin{align}
\langle|a_{n}|^{2}\rangle_{\COE} & =1+\frac{n(E-n)}{E+1}\ ,\label{eq:RMT coeffs of char poly}\\
\langle|a_{n}|^{2}\rangle_{\CUE} & =1\ .
\end{align}
These correspond to predictions for quantum chaotic systems with and
without time-reversal symmetry respectively.   In our case the directed scattering matrices break time-reversal symmetry.  Hence directed $q$-nary graphs would be in the class of systems to be modeled by the CUE.  However, as is shown
in the following, the coefficients of the characteristic polynomial
of a quantum graph are seen to deviate from the random matrix predictions
even when other spectral-statistics match the corresponding random matrix
ensemble.


Starting from \eqref{eq:a_n_expansion}, we note that $a_{0}=1$ and
averaging over $k$ the other coefficients have mean value zero,
as the average over $k$ of $\ue^{\ui kl_{\ppo}}$ is zero for pseudo
orbits of topological length $n\geq1$. The \emph{variance} of coefficients
of the characteristic polynomial was investigated numerically for the complete graph with four vertices in Ref.\onlinecite{KS99}
and also for binary graphs numerically and theoretically in Refs. \onlinecite{T00,T01}.
The approach we take here extends this discussion to the families of $q$-nary graphs, for which we obtain analytic results.
Following
\eqref{eq:a_n_expansion} we write the variance of the coefficients
of the characteristic polynomial as a sum over pairs of primitive
pseudo orbits $\ppo,\ppo'$ of the same metric length,
\begin{align}
\langle|a_{n}|^{2}\rangle_{k} & =\sum_{\ppo,\ppo'|E_{\ppo}=E_{\ppo'}=n}(-1)^{m_{\ppo}+m_{\ppo'}}A_{\ppo}\bar{A}_{\ppo'}\lim_{K\to\infty}\frac{1}{K}\int_{0}^{K}\ue^{\ui k(l_{\ppo}-l_{\ppo'})}\ud k\\
 & =\sum_{\ppo,\ppo'|E_{\ppo}=E_{\ppo'}=n}(-1)^{m_{\ppo}+m_{\ppo'}}A_{\ppo}\bar{A}_{\ppo'}\,\delta_{l_{\ppo},l_{\ppo'}}\ .\label{eq:a_n variance}
\end{align}
When the set of edge lengths is incommensurate, i.e. linearly independent
as real numbers over the rationals, the condition that
the metric lengths of the pseudo orbits be equal requires that $\ppo$
and $\ppo'$ traverse the same edges the same number of times. Then,
in the absence of time-reversal symmetry, the first order contribution
to the variance is generated by pairing an pseudo orbit with itself,
$\ppo'=\ppo$, as in the diagonal approximation of Berry \cite{B85}.
We thus define
\begin{equation}
\langle|a_{n}|^{2}\rangle_{\textrm{diag}}=\sum_{\ppo|E_{\ppo}=n}|A_{\ppo}|^{2}\ .\label{eq:a_n variance diag}
\end{equation}
From \eqref{eq:DFT_vertex_scattering} we have that the transition
probability from any incoming edge $e'$ to any outgoing edge $e$
is always,
\begin{equation}
|\sigma_{e,e'}^{(v)}|^{2}=\frac{1}{q}\
\end{equation}
and substituting in \eqref{eq:a_n variance diag} produces,
\begin{equation}
\langle|a_{n}|^{2}\rangle_{\textrm{diag}}=\sum_{\ppo|E_{\ppo}=n}\frac{1}{q^{n}}\ .
\end{equation}
Evaluating this amounts to counting the number of primitive pseudo
orbits of topological length $n$.  Then applying Corollary \ref{cor: no. of primitive pseudo orbits}
we see the diagonal approximation to the coefficients of the characteristic
polynomial of families of $q$-nary graphs is,
\begin{equation}
\langle|a_{n}|^{2}\rangle_{\textrm{diag}}=\frac{(q-1)}{q}\ .\label{eq:a_n variance diag k-ary}
\end{equation}

According to the Bohigas-Giannoni-Schmidt conjecture one might expect, in the absence of time-reversal
symmetry, the random matrix result $\langle|a_{n}|^{2}\rangle_{\CUE}=1$.
Hence the diagonal approximation deviates from this result for each family of $q$-nary graphs in the semiclassical limit, which
for graphs is the limit of large graphs, i.e fixing $q$ and taking the length of the words to infinity.  However, the discrepancy is consistent with the results for binary graphs obtained by Tanner \cite{T02}.  There
$\langle|a_{n}|^{2}\rangle_k$ is seen to converge numerically to a constant value of $0.5$ independent of $n$.  The diagonal approximation considered here reproduces this result.  
To avoid the approximation, higher order contributions to the variance of the coefficients would come from correlations between pseudo orbits of the same length with self-intersections such as the figure of eight periodic orbits considered by Seiber and Richter \cite{S02,SR01}.

To summarize, the diagonal approximation for the pseudo orbit expansion shows a deviation from  random matrix theory which does not disappear in the semi-classical limit for fixed $q$.  However, this deviation would vanish for a sequence of quantum graphs with increasing degree, i.e. increasing $q$, which is another way of approaching the semi-classical limit.   This suggests that random matrix results for the coefficients of the characteristic polynomial may be recovered for sequences of quantum graphs although under stronger conditions than those typically required for other spectral statistics such as the form factor.


\begin{acknowledgements}
The authors would like to thank Gregory Berkolaiko, Chris Joyner, Rom Pinchasi and Uzy Smilansky for helpful discussions.  This work was partially supported by a grant from the Simons Foundation (354583 to Jon Harrison).  R.B. was supported by ISF (Grant No. 494/14), Marie Curie Actions (Grant No. PCIG13-GA-2013-618468) and the Taub Foundation (Taub Fellow).\end{acknowledgements}

\end{document}